\documentclass[12pt,a4paper]{article}

\usepackage{flonat-wp}
\usepackage[harvard]{flonat-bib}

\newcommand{\deff}{\delta_{\mathrm{eff}}}

\title{Scaling, Lock-In, and Proxy Compliance:\\A Political Economy of Responsible AI}
\author{Florian A. D. Burnat\thanks{University of Bath, Bath, UK. Email: fadb20@bath.ac.uk.}
\and Brittany I. Davidson\thanks{University of Bath, Bath, UK. Email: bid23@bath.ac.uk.}}
\date{30 July 2026}

\begin{document}

\begin{titlepage}
    \maketitle
    \begin{abstract}
AI accountability at scale is an institutional problem: who can observe, verify, and change deployed systems. We develop a sequential political-economy model in which an AI vendor chooses auditability and substantive mitigation, a deployer monitors after adoption while facing switching costs, and enforcement depends on verifiable evidence. Anticipating the deployer's monitoring response, the vendor may stop at an observable procurement floor while mitigating below the social first best, producing a \emph{proxy-compliance equilibrium}. We characterize the unique interior equilibrium and the corner in which harm is fully mitigated. Independent audit rights raise enforcement exposure directly; portability restores deployer leverage; incident reporting adds a regulator-visible evidence channel; and outcome-linked liability creates incentives that do not depend on vendor-controlled detection. The results explain why documentation and standardized evaluations can coexist with persistent post-deployment harms, and generate testable implications for monitoring, mitigation, and the gap between formal compliance and operational outcomes.

    \end{abstract}
    \setcounter{page}{0}
    \thispagestyle{empty}
\end{titlepage}

\clearpage
\doublespacing

\section{Introduction}
\label{sec:intro}

\paragraph{Motivating puzzle.}
Responsible AI checklists, model cards, conformity documentation, and standardized evaluations have become common; however, recurring harms and accountability gaps remain. Four observations sharpen this puzzle. \emph{First}, the EU AI Act places extensive documentation and risk-management obligations on providers, but many relevant inputs---training data, evaluation protocols, red-team logs, and access conditions---remain under provider control \parencite{Veale2021-uu,Schuett2024-am}. \emph{Second}, independent reproduction of vendor-reported safety and fairness claims is often infeasible under black-box or contractually restricted access \parencite{Mitchell2019-ee,Mokander2024-hd,Casper2024-em}. \emph{Third}, audits can be scoped around negotiated data, metrics, and publication rights rather than independently imposed tests \parencite{Raji2019-xg,Costanza-Chock2022-sn}. \emph{Fourth}, once a deployer has integrated a model into workflows, fine-tuning pipelines, prompt libraries, and downstream tooling, switching providers can be costly. Evidence from frontier-AI firms' voluntary commitments is consistent with the resulting gap between visible reporting and underlying performance \parencite{Wang2025-tn}. The puzzle is therefore not merely why organizations fail to implement an otherwise complete ethics framework, but why visible compliance and limited mitigation may coexist as rational choices.

\paragraph{Thesis.}
We argue that proxy compliance can arise when enforcement requires verifiable evidence, the vendor controls part of the evidence-generating infrastructure, and post-integration switching costs weaken the deployer's recourse. Lock-in reduces monitoring; weaker monitoring lowers verifiable detection; and lower enforcement exposure reduces the vendor's return to substantive mitigation. Lock-in is an amplifier rather than a necessary primitive; constrained independent detection can sustain the mechanism even when switching costs are moderate.

\paragraph{Contributions.}
First, we develop a vendor--deployer model in which auditability, mitigation, and monitoring are endogenous, conditional on institutional parameters for procurement, switching costs, and enforcement. Second, we solve the sequential game while allowing the vendor to anticipate how mitigation changes the deployer's monitoring response. This yields a unique interior equilibrium, a transparent full-mitigation corner, and policy comparative statics. Third, we connect the formal levers to current governance institutions and specify empirical measures and settings in which the mechanism can be tested. The contribution is a tractable institutional explanation of the gap between observable responsibility signals and harm reduction, not a claim that any single governance instrument is sufficient. In this paper, scale is the institutional setting rather than an endogenous variable; the formal analysis isolates how auditability, switching costs, and evidence-dependent enforcement shape mitigation once AI systems are embedded in organizational workflows.

\paragraph{Roadmap.}
Section~\ref{sec:related} positions the argument. Section~\ref{sec:model} introduces the model.
Section~\ref{sec:analysis} characterizes the equilibrium. Section~\ref{sec:policy} studies policy levers and welfare.
Section~\ref{sec:pred} develops empirical implications, and Section~\ref{sec:discussion} discusses the scope and limitations.

\begin{table*}[t]
    \centering
    \caption{Main results and policy interpretation. Formal derivations are in the appendices.}
    \label{tab:results}
    \small
        \begin{tabular}{@{}p{0.09\textwidth}p{0.42\textwidth}p{0.43\textwidth}@{}}
        \toprule
        Result & Formal claim & Institutional interpretation \\
        \midrule
        Prop.~\ref{prop:proxy} & weak baseline detection plus high lock-in yields proxy compliance
         & visible compliance can coexist with low mitigation \\
        Prop.~\ref{prop:account} & more independent detection or less lock-in raises mitigation
         & accountability needs evidence and credible recourse \\
        Thm.~\ref{thm:audit} & sufficiently strong audit rights work at any lock-in level
         & independent evaluation can bypass deployer weakness \\
        Prop.~\ref{prop:portability} & portability raises mitigation through monitoring
         & standards restore deployer leverage \\
        Prop.~\ref{prop:incident} & incident reporting adds an independent evidence channel
         & credible reports reduce dependence on vendor-gated access \\
        Prop.~\ref{prop:liability} & outcome-linked exposure raises mitigation independently of detection
         & liability retains bite when verification is weak \\
        \bottomrule
    \end{tabular}
\end{table*}

\section{Related Work}
\label{sec:related}

This argument relates to responsible AI governance, algorithmic auditing, the economics of switching costs and lock-in, and the political economy of regulation under asymmetric information.

\paragraph{Responsible AI as institutional practice.}
A growing body of work documents the gap between Responsible AI principles and practice. \textcite{Jobin2019-pk} survey 84 AI ethics guidelines, finding convergence on high-level principles but divergence in interpretation and a near-universal absence of enforcement mechanisms. This landscape has been characterized as ``ethics washing''---voluntary frameworks that create reputational benefits without constraining behavior \parencite{Hagendorff2020-vs}. \textcite{Raji2022-pe} sharpen the diagnosis by cataloging functionality failures that no current audit regime reliably detects, and ML research incentives empirically select against harm-reduction methodology \parencite{Birhane2022-sy}---both findings consistent with a system in which signals of responsibility are easier to produce than the underlying substance. \textcite{Crawford2022-wt} maps the political economy of AI production, documenting how concentration of compute, data, and infrastructure creates the power asymmetries that enable this dynamic. At the organizational level, internal accountability practices are shaped more by what is demonstrable to external audiences than by what effectively reduces harm \parencite{Raji2020-lq, Moss2021-xq}. \textcite{Selbst2019-pa} identify ``abstraction traps''---when responsible AI is operationalized through portable, formalized metrics detached from deployment context, it produces the appearance of compliance without addressing situated harms. In our model, proxy compliance is the equilibrium where observable signals are cheaper than substantive mitigation and enforcement depends on verifiable evidence.

\paragraph{Signaling and corporate social responsibility.}
The mechanism we label \emph{proxy compliance} has a well-developed formal ancestor in the economics of signaling and corporate social responsibility. \textcite{Benabou2010-ml} decompose visibly responsible behavior into intrinsic, material, and reputational-signaling motives, and show that in equilibrium the observable action can dominate the unobservable effort when audiences condition on the action alone---even when the two are conceptually distinct. This is the structure we specialize: when deployers, regulators, and civil society condition on \emph{observable} documentation and standardized evaluations, vendors face an incentive problem in which producing the visible signal is strategically separable from substantive harm reduction. The contribution relative to \textcite{Benabou2010-ml} is domain-specific: we embed the signaling mechanism in a vendor--deployer--regulator game with switching costs and evidence-based enforcement, and derive comparative statics for policy levers that are specific to the AI value chain (audit rights, portability, incident reporting, and outcome-linked liability) rather than generic CSR. Recent empirical work on frontier-AI firms' voluntary White House commitments documents widespread partial compliance and reporting-via-marketing \parencite{Wang2025-tn}---a pattern our equilibrium predicts when audit rights and portability are weak.

\paragraph{Audits, evaluation, and verifiability.}
Algorithmic auditing has been proposed as a key accountability mechanism \parencite{Sandvig2014-kw, Raji2019-xg}; however, it faces structural limitations when vendors control access. \textcite{Mokander2022-ul} provide a taxonomy of audit types that underwrites our auditability variable $t$: their distinction between vendor self-assessment and third-party conformity assessment maps directly onto the detection function~$p(t,m)$. The auditing ecosystem itself lacks standardization, independence guarantees, and enforcement power \parencite{Costanza-Chock2022-sn}; \textcite{Floridi2022-ax} propose an operational conformity-assessment procedure (capAI) whose feasibility rests on the same vendor cooperation that our model treats as endogenous. For large language models, auditability constraints are severe \parencite{Mokander2024-hd}: model opacity, generality, and absent harm benchmarks. Documentation standards such as model cards \parencite{Mitchell2019-ee} are necessary but insufficient---they are only as reliable as the information vendors choose to include, making them vulnerable to strategic disclosure. \textcite{Wachter2021-qi} provide the doctrinal ground for this limitation: automated compliance checks are evidentiarily incomplete substitutes for contextual legal assessment, which justifies treating $p(t,m)$ as a proxy rather than a truth function. In our framework, these observations correspond to the detection probability $p(t,m)$ being substantially vendor-determined, giving vendors scope to limit auditability and thereby constrain enforcement. \textcite{Casper2024-em} show that black-box API access is \emph{insufficient} for rigorous audits of frontier systems---internal access to weights, training data, and evaluation logs is required for safety and fairness claims to be credibly tested. Consistent with this, \textcite{Burnat2026-xq} document concrete ``audit blind-spots'' across major platforms (X, Reddit, TikTok, Meta) where API restrictions block the independent verification that transparency mandates presume---an empirical instance of vendor-controlled $p(t,m)$. \textcite{Raji2022-ks} argue on this basis for a designed third-party audit ecosystem with enforcement powers, and \textcite{Schuett2025-je} argues for a standing internal-audit function inside frontier-AI firms with direct board access. Our model treats the auditability variable $t$ as precisely the lever these authors identify: moving from vendor-gated black-box evaluation toward independent internal access raises $p(t,m)$ at the margin and, under conditions we characterize, tips the equilibrium from proxy compliance to accountability.

\paragraph{Market structure, lock-in, and bargaining power.}
The economics of switching costs provides the foundation for our lock-in mechanism. Switching costs give firms ex post market power over locked-in customers \parencite{Klemperer1995-zn}, and dynamic competition with switching costs leads to an ``invest then harvest'' pattern \parencite{Farrell1988-rs}. Network externalities create self-reinforcing concentration \parencite{Katz1985-zg}, which in AI manifests as ecosystem lock-in around dominant foundation-model providers; foundation-model scale concentrates development capacity among a few well-resourced actors \parencite{Bender2021-ct}, supplying the market-structure precondition for the switching cost~$s$ we formalize. ``Openness'' in AI is often strategically deployed to consolidate rather than democratize power \parencite{Gray-Widder2023-vv}, and regulatory compliance costs can paradoxically reinforce concentration \parencite{Engler2023-rq}. Procurement frameworks are poorly suited to evaluating AI systems \parencite{Mulligan2019-yt}, weakening deployer-side discipline even before lock-in takes effect. Our model integrates these insights: switching cost~$s$ dampens effective recourse from monitoring, removing the market discipline that would otherwise incentivize vendor mitigation.

\paragraph{Evidence-based enforcement.}
The political economy of regulation under asymmetric information is well established: regulation may be captured by industry \parencite{Stigler1971-uz}; principal-agent rents distort outcomes when information is asymmetric \parencite{Laffont1993-cc}; and optimal enforcement institutions depend on the power asymmetry between regulators and firms \parencite{Glaeser2003-to}. Intermediary auditors can themselves be captured, reducing rather than raising oversight effectiveness \parencite{Tirole1986-xu}. The choice between liability and safety regulation depends on information structures \parencite{Shavell1984-mf}---directly relevant to our comparison of evidence-based enforcement versus outcome-linked liability. In the AI context, the Act's heavy reliance on vendor self-assessment and harmonized-standard compliance is well documented, as is its broader risk-management architecture \parencite{Veale2021-uu,Schuett2024-am}. Subsequent 2024--2025 scholarship has identified three structural weaknesses. First, scholarship on the AI Office, AI Board, and Scientific Panel flags under-resourcing relative to the Act's evidentiary burden \parencite{Hacker2023-nc,Novelli2024-rt}. Second, harmonized-standard delegation routes the regulator's burden of proof through vendor-influenced standards bodies \parencite{Gornet2024-ie,Laux2024-eg}. Third, the Fundamental Rights Impact Assessment adds a rights-based layer above conformity assessment, and incorporating civil society into the audit ecosystem can guard against capture \parencite{Mantelero2024-cg,Hartmann2024-ob}. The NIST AI Risk Management Framework's generative-AI profile \parencite{National-Institute-of-Standards-and-Technology-US-2024-am} is now the operational benchmark against which AI risk management is graded in practice. Pulling these threads together, evidence-based enforcement makes penalties contingent on detection probability, so vendors face an incentive to limit transparency whenever the marginal adoption benefit from auditability falls below the marginal increase in expected penalties.


\section{Model}
\label{sec:model}

The model separates \emph{visible compliance} from \emph{substantive mitigation}. Auditability~$t$ helps the vendor satisfy procurement requirements but also facilitates enforcement; mitigation~$e$ directly reduces harm. The deployer can monitor after adoption, but monitoring is valuable only when evidence creates credible recourse. Switching cost~$s$ weakens that recourse. Independent audit capacity~$\Delta$, reporting-generated exposure~$\mu\rho$, and outcome-linked exposure~$\lambda_O F_O$ act through different institutional channels (Figure~\ref{fig:mechanism}).

\begin{figure}[t]
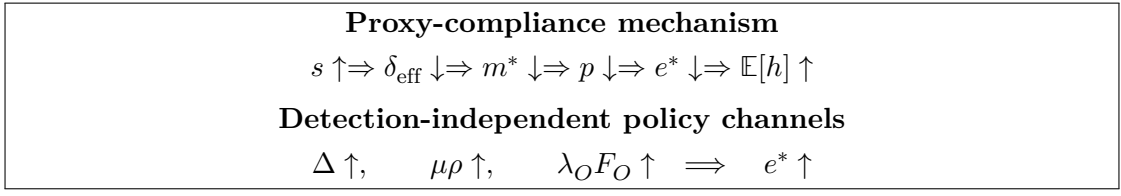

    \centering
    \fbox{\parbox{0.91\columnwidth}{\centering\small
    \textbf{Proxy-compliance mechanism}\\[3pt]
    $s\uparrow \Rightarrow \deff\downarrow \Rightarrow m^\ast\downarrow
    \Rightarrow p\downarrow \Rightarrow e^\ast\downarrow \Rightarrow \E[h]\uparrow$
    \\[6pt]
    \textbf{Detection-independent policy channels}\\[3pt]
    $\Delta\uparrow,\qquad \mu\rho\uparrow,\qquad \lambda_O F_O\uparrow
    \quad\Longrightarrow\quad e^\ast\uparrow$
    }}
    \caption{Lock-in weakens the deployer-monitoring channel. Audit rights~$\Delta$, effective incident-reporting
    coverage~$\rho$ with sanction scale~$\mu$, and outcome-linked exposure~$\lambda_O F_O$ add enforcement channels
    that the vendor cannot suppress through ordinary auditability choices. Portability works by lowering~$s$.}
    \label{fig:mechanism}
\end{figure}

\subsection{Players and Timeline}

An institutional environment fixes the procurement floor~$\bar t$, evidence-based penalty scale~$\lambda F$, independent audit capacity~$\Delta$, and switching-cost environment~$s$. The strategic subgame is:

\begin{enumerate}[leftmargin=*,itemsep=2pt,topsep=2pt]
    \item The vendor chooses auditability $t\in[0,1]$ and mitigation $e\ge0$.
    \item The deployer observes $(t,e)$, chooses adoption $a\in\{0,1\}$, and if adopting chooses monitoring $m\ge0$.
    Integration creates switching costs represented by~$s$.
    \item Harm $h\in\{0,1\}$ is realized. Evidence is generated with probability $p(t,m)$, and the regulator applies the
    specified enforcement rule when actionable evidence exists.
\end{enumerate}

\begin{table}[t]
    \centering
    \caption{Key notation.}
    \label{tab:notation}
    \small
    \begin{tabular}{@{}p{0.22\columnwidth}p{0.69\columnwidth}@{}}
        \toprule
        Symbol & Meaning \\
        \midrule
        $t,e,m$ & vendor auditability, vendor mitigation, deployer monitoring \\
        $s$ & switching cost; $\deff=\delta/(1+s)$ is effective recourse \\
        $q(\theta,e)$ & harm probability \\
        $p(t,m)$ & vendor/deployer evidence-generation probability \\
        $\Delta$ & independent audit exposure, added to $p$ \\
        $\lambda F$ & evidence-based penalty scale \\
        $\mu\rho$ & reporting-generated exposure ($\rho$: effective coverage) \\
        $\lambda_O F_O$ & outcome-linked exposure \\
        $\bar t$ & procurement/compliance floor \\
        $\kappa$ & deployer harm exposure \\
        $\alpha,\beta$ & sensitivity of detection to $t,m$ \\
        $k_t,k_e,k_m$ & quadratic cost coefficients \\
        \bottomrule
    \end{tabular}
\end{table}

\subsection{Risk, Detection, and Payoffs}

We condition on a common-knowledge risk parameter $\theta\in(0,1)$. Mitigation reduces the probability of harm, $\Prob(h=1\mid\theta,e)=q(\theta,e)$, and actionable evidence arises with probability $p(t,m)$, which increases in both auditability and monitoring.

The vendor earns adoption benefit $B>0$ and incurs convex auditability and mitigation costs:
\begin{equation}
\label{eq:vendor}
    U_V=aB-c_t(t)-c_e(e)-\lambda F\,q(\theta,e)\,[p(t,m)+\Delta].
\end{equation}
Here $p+\Delta$ is an enforcement-exposure index rather than necessarily a literal probability. The deployer obtains value $V$, pays integration cost $I$ and monitoring cost $c_m(m)$, and can recover a fraction~$\delta$ of harm exposure when evidence creates recourse. Lock-in reduces effective recourse to $\deff=\delta/(1+s)$:
\begin{equation}
\label{eq:deployer}
    U_D=a\!\left[V-I-c_m(m)-\kappa q(\theta,e)\bigl(1-\deff p(t,m)\bigr)\right].
\end{equation}
The regulator is represented through the institutional parameters $(\lambda F,\Delta)$; Section~\ref{sec:policy} asks how changing those parameters alters equilibrium.

\subsection{Tractable Specification}

\begin{assumption}[Functional forms]
\label{ass:mvp}
    For the comparative statics,
    \[
        \begin{aligned}
            q(\theta,e)&=\max\{0,\theta-e\},\\
            p(t,m)&=1-\exp(-\alpha t-\beta m),
        \end{aligned}
    \]
    and $c_j(x)=k_jx^2/2$ for $j\in\{t,e,m\}$, with all cost and sensitivity parameters positive.
\end{assumption}

\begin{assumption}[Procurement threshold]
\label{ass:adoption}
    Adoption occurs if and only if $t\ge\bar t$. The floor~$\bar t$ is an exogenous procurement or certification     requirement. We condition on adoption and state the vendor participation condition in Appendix~\ref{app:equilibrium}.
\end{assumption}

We use subgame-perfect equilibrium: the deployer chooses monitoring after observing vendor actions, and the vendor anticipates that monitoring response.

\section{Equilibrium Analysis}
\label{sec:analysis}

\subsection{Deployer Monitoring}

\begin{lemma}[Deployer best response]
\label{lem:deployer}
    For $q(\theta,e)>0$ and $\deff>0$, the deployer has a unique interior monitoring choice~$m^\ast$ satisfying
    \begin{equation}
    \label{eq:foc_m}
        k_m m^\ast=\kappa q(\theta,e)\deff\beta
        \exp(-\alpha t-\beta m^\ast).
    \end{equation}
    Monitoring decreases with switching cost~$s$ and increases with deployer harm exposure~$\kappa$.
\end{lemma}

\begin{proof}[Proof sketch]
    Marginal monitoring cost is increasing and marginal recourse is decreasing, giving a unique crossing. Implicit differentiation of~\eqref{eq:foc_m} yields the signs. Appendix~\ref{app:equilibrium} gives the derivatives.
\end{proof}

Lock-in makes information less useful: when the deployer's exit or compensation threat is not credible, evidence has less strategic value and monitoring falls.

\subsection{Vendor Choices}

Define the vendor's marginal enforcement exposure after anticipating the monitoring response:
\begin{equation}
\label{eq:H}
    \begin{aligned}
        H(t,m)&\equiv p(t,m)+(1-p(t,m))\frac{\beta m}{1+\beta m}\\
        &=1-\frac{\exp(-\alpha t-\beta m)}{1+\beta m}.
    \end{aligned}
\end{equation}

\begin{lemma}[Vendor best response]
\label{lem:vendor}
    Suppose adoption benefit~$B$ is large enough that the vendor meets the procurement floor.
    \begin{enumerate}[label=(\alph*),itemsep=2pt]
        \item The vendor chooses $t^\ast=\bar t$.
        \item In the interior region $e^\ast<\theta$, mitigation satisfies
        \begin{equation}
        \label{eq:foc_e}
            e^\ast=\frac{\lambda F}{k_e}\,[H(\bar t,m^\ast)+\Delta].
        \end{equation}
        If the right-hand side reaches~$\theta$, mitigation clamps at $e^\ast=\theta$ and harm is zero.
    \end{enumerate}
\end{lemma}

\begin{proof}[Proof sketch]
    Higher $t$ beyond the procurement floor raises both cost and total detection, even after the induced fall in monitoring. For mitigation, implicit differentiation of~\eqref{eq:foc_m} gives $\partial m^\ast/\partial e=-m^\ast/[(\theta-e)(1+\beta m^\ast)]<0$. Substitution adds $(1-p)\beta m^\ast/(1+\beta m^\ast)$ to the direct detection term in~$H$, yielding~\eqref{eq:foc_e}: the monitoring response is an \emph{additional private return} to mitigation. The second derivative is negative; Appendix~\ref{app:equilibrium} gives the full derivation.
\end{proof}

Therefore, the vendor satisfies the visible requirement and stops increasing auditability, while choosing mitigation according to the expected enforcement exposure created by both auditability and deployer monitoring.

\subsection{Proxy Compliance and Accountability}

\begin{definition}[Proxy compliance]
\label{def:proxy}
    An equilibrium exhibits \emph{proxy compliance} when $t^\ast=\bar t$ but $e^\ast<e^{FB}$, where the social first-best mitigation is
    \[
        e^{FB}=\min\{1/k_e,\theta\}.
    \]
\end{definition}

The equilibrium is characterized by
\begin{align}
    t^\ast&=\bar t, \label{eq:t_star}\\
    e^\ast&=\min\!\left\{\theta,\frac{\lambda F}{k_e}[H(\bar t,m^\ast)+\Delta]\right\}, \label{eq:e_star}\\
    k_m m^\ast&=\kappa(\theta-e^\ast)\deff\beta
    \exp(-\alpha\bar t-\beta m^\ast), \label{eq:m_star}
\end{align}
with $m^\ast=0$ at the full-mitigation corner. In the interior, the right-hand side of~\eqref{eq:e_star} is strictly decreasing in~$e$ through the monitoring response, so the equilibrium is unique.

\begin{proposition}[Existence of proxy compliance]
\label{prop:proxy}
    Suppose $\Delta=0$ and baseline enforcement is weak enough that
    \[
        \frac{\lambda F}{k_e}\bigl(1-\exp(-\alpha\bar t)\bigr)<e^{FB}.
    \]
    Then there exists $\bar s$ such that sufficiently high switching cost yields proxy compliance. As $s\to\infty$, $m^\ast\to0$ and
    \[
        e^\ast\to\frac{\lambda F}{k_e}\bigl(1-\exp(-\alpha\bar t)\bigr).
    \]
\end{proposition}

\begin{proposition}[Accountability under stronger institutions]
\label{prop:account}
    Within the interior equilibrium: (a) $de^\ast/d\Delta>0$, although monitoring feedback attenuates the direct effect; (b) $de^\ast/ds<0$, so reducing switching costs raises mitigation; and (c) sufficiently strong independent detection reaches $e^{FB}$, with portability reinforcing the transition.
\end{proposition}

\subsection{Worked Example}
\label{sec:worked}

Let $\theta=0.6$, $k_t=k_e=k_m=1$, $\alpha=\beta=\kappa=\delta=\lambda F=1$, and $\bar t=0.2$. Then $e^{FB}=0.6$. Solving~\eqref{eq:e_star}--\eqref{eq:m_star} yields Table~\ref{tab:numerical}.

\begin{table}[t]
    \centering
    \caption{Equilibrium values in three institutional settings.}
    \label{tab:numerical}
    \small
    \begin{tabular}{@{}lccccc@{}}
        \toprule
        Setting & $s$ & $\Delta$ & $m^\ast$ & $e^\ast$ & $\E[h]$ \\
        \midrule
        1. High lock-in, no audit & 9 & 0 & 0.03 & 0.23 & 0.37 \\
        2. Portability only & 1 & 0 & 0.10 & 0.33 & 0.27 \\
        3. Combined reform & 1 & 1 & 0.00 & 0.60 & 0.00 \\
    \bottomrule
    \end{tabular}
\end{table}

Portability alone produces a meaningful but incomplete improvement: Setting~2 moves above half of the first best but does not reach accountability. Under combined reform, the desired interior mitigation exceeds~$\theta$; harm is eliminated and the deployer no longer monitors because there is no residual harm to investigate.

\begin{figure}[t]
    \centering
    \includegraphics[width=\columnwidth]{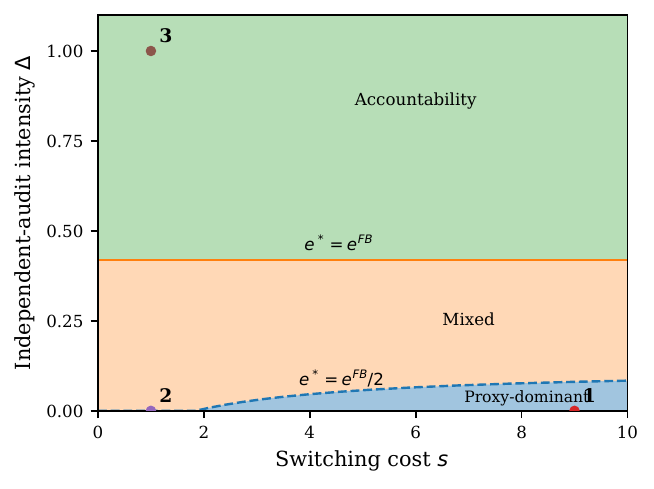}
    \caption{Illustrative regimes in the $(s,\Delta)$ plane under the worked-example parameters. Proxy compliance formally includes every point below first best; for visualization, the lower region is labeled     \emph{proxy-dominant} when $e^\ast<e^{FB}/2$, the middle region is mixed, and the upper region reaches $e^{FB}$. Markers correspond to Table~\ref{tab:numerical}.}
    \label{fig:regimes}
\end{figure}

\subsection{Comparative Statics}

The total equilibrium effects differ from partial effects because policy-induced mitigation changes monitoring. Higher $s$ lowers both monitoring and mitigation; higher $\Delta$, $\lambda F$, or $\bar t$ raises mitigation and reduces residual monitoring; higher mitigation cost~$k_e$ has the opposite effect; and higher deployer harm exposure $\kappa$ raises monitoring and thereby mitigation. Appendix~\ref{app:comparative} provides derivatives and a summary table.

\section{Policy Levers and Welfare}
\label{sec:policy}

\subsection{Audit Rights and Independent Evaluation}

Independent audit capacity raises enforcement exposure by~$\Delta$ without relying on deployer monitoring or ordinary vendor-controlled access.

\begin{theorem}[Audit rights eliminate proxy compliance]
\label{thm:audit}
    For every fixed lock-in level~$s$ and $\lambda F>0$, a finite audit threshold~$\bar\Delta(s)$ exists such that $\Delta\ge\bar\Delta(s)$ yields $e^\ast\ge e^{FB}$. A uniform sufficient condition is
    \[
        \Delta\ge \frac{k_e e^{FB}}{\lambda F}.
    \]
    Thus sufficiently strong independent detection eliminates proxy compliance at any lock-in level.
\end{theorem}

The exact threshold is lower because baseline detection and the monitoring-response term also contribute to mitigation. Audit rights are therefore the strongest single lever in the model, but their practical effect depends on genuinely independent access and credible enforcement rather than the formal existence of an audit process.

\subsection{Portability and Standards}

\begin{proposition}[Portability raises mitigation]
\label{prop:portability}
    Reducing switching cost~$s$ raises equilibrium monitoring and mitigation and lowers expected harm. The effect is stronger when deployer harm exposure~$\kappa$ and monitoring productivity~$\beta$ are high.
\end{proposition}

Portability works indirectly: compatible interfaces, exportable artifacts, and credible re-tendering make exit and recourse more credible, raising the value of monitoring and, therefore, the vendor's mitigation incentive.

\subsection{Incident Reporting}

Let $\rho\in[0,1]$ denote the effective coverage of a mandatory reporting channel: the probability that a reportable incident reaches the regulator in actionable form. Let $\mu>0$ be the associated enforcement-exposure scale. The vendor then faces additional expected exposure~$\mu\rho\,q(\theta,e)$.

\begin{proposition}[Incident reporting substitutes for vendor-gated evidence]
\label{prop:incident}
    Under effective incident reporting, the mitigation condition becomes
    \begin{equation}
    \label{eq:e_report}
        e_R^\ast=\min\!\left\{\theta,\frac{\lambda F[H(\bar t,m_R^\ast)+\Delta]+\mu\rho}{k_e}\right\}.
    \end{equation}
    Hence $e_R^\ast>e^\ast$ whenever $\mu\rho>0$ and the baseline is below the mitigation corner.
\end{proposition}

The product~$\mu\rho$ is deliberately an \emph{effective} channel: a nominal reporting duty with narrow scope, under-reporting, or weak follow-up contributes little.

\subsection{Outcome-Linked Liability}

Let $\lambda_OF_O$ denote the exposure tied directly to realized harm rather than to the separate evidence-generation channel~$p(t,m)$.

\begin{proposition}[Outcome-linked exposure raises mitigation]
\label{prop:liability}
    The mitigation condition becomes
    \begin{equation}
    \label{eq:e_liability}
        e_L^\ast=\min\!\left\{\theta,\frac{\lambda F[H(\bar t,m_L^\ast)+\Delta]+\lambda_OF_O}{k_e}\right\},
    \end{equation}
    so any positive outcome-linked exposure raises mitigation below the corner. Its direct component is unaffected by vendor-controlled auditability.
\end{proposition}

This is a stylized benchmark, not a claim that existing liability laws automatically observe or compensate for every harm.

\subsection{Welfare}

Social welfare equals adoption value minus expected harm and real resource costs. Below first best, an increase in independent audit exposure raises mitigation and lowers both harm and residual monitoring; therefore, its total welfare effect is positive. Portability has an additional trade-off: it raises beneficial mitigation but can also raise costly monitoring. The precise condition is stated in Appendix~\ref{app:welfare}.

\subsection{Mapping to Current EU Regulation}
\label{sec:mapping}

The AI Act combines provider documentation and risk-management duties with differentiated conformity-assessment procedures. Article~43 and Annex~VI permit internal-control routes for relevant classes of high-risk systems, while Articles~53 and~55 impose obligations on general-purpose model providers and systemic-risk models; the AI Office has information and evaluation powers under Articles~91--93 \parencite{EuropeanUnion2024AIAct}. In the model, documentation that remains provider-controlled primarily affects~$t$, whereas regulator-controlled evaluation capacity raises~$\Delta$.

Article~40 of the Digital Services Act (DSA) provides an analogous $\Delta$-raising design for very large online platforms and search engines \parencite{EuropeanUnion2022DSA}. Portability and interoperability duties under the Digital Markets Act (DMA) illustrate instruments that can lower switching costs, although they are not an AI-specific portability regime \parencite{EuropeanUnion2022DMA}. AI Act Article~73 serious-incident reporting can create a $\mu\rho$ channel when reports are complete, actionable, and credibly enforced.

The proposed AI Liability Directive was a fault-based evidentiary instrument centered on disclosure and rebuttable presumptions rather than the outcome-linked benchmark in Proposition~\ref{prop:liability}; it was formally withdrawn in October~2025 \parencite{EuropeanCommission2025Withdrawal}. The revised Product Liability Directive covers software and provides no-fault liability for defective products but still requires defect, damage, and causation \parencite{EuropeanUnion2024PLD}. It is therefore closer to, but not identical with, the model's direct outcome-exposure channel.

\section{Predictions and Measurement Implications}
\label{sec:pred}

Because $t^\ast=\bar t$, the model's direct comparative statics concern monitoring, mitigation, and harm; claims about contract terms or documentation are institutional hypotheses motivated by the mechanism.

\subsection{Predictions}

\begin{enumerate}[leftmargin=*,itemsep=2pt,topsep=2pt]
    \item \textbf{Lock-in weakens monitoring and mitigation.} Greater integration depth and fewer credible alternatives predict lower monitoring, lower effective detection, and lower mitigation.
    \item \textbf{Visible compliance can remain stable while outcomes deteriorate.} As $s$ rises, auditability remains pinned to the procurement floor, while mitigation falls, producing documentation without commensurate reductions in incident rates.
    \item \textbf{Independent access shifts operational effort.} Exogenous increases in regulator or third-party evaluation capacity should narrow the gap between vendor-reported and independently measured performance and raise mitigation investment.
    \item \textbf{Portability is sector-dependent.} Its effect should be largest where deployer harm exposure is high and monitoring technology is mature; audit rights should bind more strongly where those conditions are absent.
\end{enumerate}

\subsection{Operationalization}

Auditability can be measured through contractual audit clauses, API evaluation terms, disclosure completeness, and third-party publication rights. Lock-in can be proxied by fine-tuned endpoints, provider-specific tool wrappers, checkpoint or prompt-format portability, retraining costs, and observed switching. Independent exposure can be measured through regulator access, evaluation budgets, researcher-access programs, and enforcement actions. Candidate settings include medical-device AI, credit scoring, public-sector hiring, DSA audit and transparency records for very large online platforms (VLOPs), provider migration and re-tender data, and incident databases supplemented by AI Act Article~73 reports. Appendix~\ref{app:empirical} develops the measurement strategy.

\section{Discussion and Limitations}
\label{sec:discussion}

\paragraph{Audit is not a panacea.}
Third-party audit can itself become proxy compliance. Audits may operate as rational myths---adopted for legitimacy rather than demonstrated efficacy---and auditors and auditees can share an interest in the appearance of rigor \parencite{Suchman1996-mm,Eckstein2024-ll}. In the model, ceremonial or captured auditing contributes little to effective $\Delta$; independence, access, publication rights, and enforcement determine whether the threshold in Theorem~\ref{thm:audit} is reached.

\paragraph{How central is lock-in?}
The mechanism does not require extreme switching costs. Lock-in matters because it disables a monitoring channel that could otherwise substitute for weak independent detection. When deployer harm exposure is low, monitoring is immature, or vendor-independent access is constrained, the system can remain near proxy compliance, even at moderate~$s$. Conversely, credible portability is powerful in high-exposure sectors with mature monitoring.

\paragraph{Modeling scope.}
The model is deliberately static and single-vendor. The title's references to ``scaling'' and ``political economy'' describe the motivating institutional environment, not endogenous scale or market structure. Switching cost, enforcement, and the procurement floor are parameters rather than outcomes of competition or regulatory choice. The common-knowledge risk benchmark omits private information; multidimensional harms, civil-society evidence production, adversarial misuse, multi-vendor competition, and repeated reputation dynamics are left to extensions. These abstractions make the policy channels transparent but limit quantitative interpretation.

\paragraph{Empirical status.}
The numerical example is illustrative rather than calibrated. The comparative statics generate hypotheses and measurement targets, not estimates of policy effect sizes. Cross-sector variation in procurement floors, audit access, and integration depth can support empirical tests; however, those variables may be correlated with underlying risk and institutional capacity; therefore, identification requires careful design rather than a simple cross-sectional comparison.

\section{Conclusion}
\label{sec:conclusion}

Responsible AI commitments can coexist with inadequate mitigation because visible compliance, verifiable evidence, and substantive harm reduction are distinct strategic goals. A vendor that controls ordinary auditability meets the procurement floor; a locked-in deployer monitors less; and evidence-dependent enforcement then supplies too little mitigation incentive. Independent audit rights bypass the weakened monitoring channel, portability restores deployer leverage, incident reporting creates regulator-visible evidence, and outcome-linked exposure retains bite when ordinary detection is weak. None are automatically effective: audit independence, reporting coverage, credible sanctions, and sector-specific monitoring capacity determine whether formal rules change equilibrium behavior. The central policy implication is therefore institutional rather than procedural: accountability requires both verifiability and power to act on what is verified.

\appendix

\section{Mathematical Derivations and Equilibrium Proofs}
\label{app:equilibrium}

Throughout this appendix, consider the interior region $0<e<\theta$ and write
\[
    E(t,m)\equiv \exp(-\alpha t-\beta m)=1-p(t,m),\qquad
A\equiv\frac{\lambda F}{k_e}.
\]

\subsection{Deployer Best Response}

\begin{proof}[Proof of Lemma~\ref{lem:deployer}]
    The deployer's first-order condition is
    \[
        \frac{\partial U_D}{\partial m}
        =-k_m m+\kappa q(\theta,e)\deff\beta E(t,m)=0,
    \]
    which is~\eqref{eq:foc_m}. Its second derivative is
    \[
        -k_m-\kappa q(\theta,e)\deff\beta^2E(t,m)<0,
    \]
    so the solution is unique. Let
    \[
        G(m;t,e,s,\kappa)
        =k_m m-\kappa(\theta-e)\deff\beta E(t,m).
    \]
    At the optimum, use~\eqref{eq:foc_m} to obtain
    \[
        G_m=k_m+\kappa(\theta-e)\deff\beta^2E
        =k_m(1+\beta m^\ast)>0.
    \]
    Implicit differentiation gives
    \begin{align}
        m_e^\ast&=-\frac{m^\ast}{(\theta-e)(1+\beta m^\ast)}<0, \label{eq:me}\\
        m_t^\ast&=-\frac{\alpha m^\ast}{1+\beta m^\ast}<0, \label{eq:mt}\\
        m_s^\ast\big|_{e,t}
        &=-\frac{m^\ast}{(1+s)(1+\beta m^\ast)}<0, \label{eq:ms-partial}\\
        m_\kappa^\ast\big|_{e,t,s}
        &=\frac{m^\ast}{\kappa(1+\beta m^\ast)}>0. \label{eq:mkappa-partial}
    \end{align}
    The first two derivatives describe how the deployer's best response changes when vendor choices change; the latter two are partial institutional effects holding vendor mitigation fixed.
\end{proof}

\subsection{Vendor Auditability and Mitigation}

\begin{proof}[Proof of Lemma~\ref{lem:vendor}]
    \emph{Auditability.} Conditional on adoption, substitute $m^\ast(t,e)$ into~\eqref{eq:vendor}. From~\eqref{eq:mt},
    \[
        \frac{d}{dt}p(t,m^\ast(t,e))
        =p_t+p_m m_t^\ast
        =\frac{\alpha E(t,m^\ast)}{1+\beta m^\ast}>0.
    \]
    Consequently,
    \[
        \frac{dU_V}{dt}
        =-k_t t-\lambda Fq(\theta,e)
        \frac{\alpha E(t,m^\ast)}{1+\beta m^\ast}<0
        \qquad(t\ge\bar t).
    \]
    Once adoption is secured, every further increase in~$t$ raises both cost and expected enforcement exposure. Therefore $t^\ast=\bar t$.
    
    \emph{Mitigation.} Using $q_e=-1$ and~\eqref{eq:me},
    \begin{align*}
        \frac{dU_V}{de}
        &=-k_e e+\lambda F[p(\bar t,m^\ast)+\Delta]
        -\lambda Fq\,p_m m_e^\ast\\
        &=-k_e e+\lambda F\Biggl[
        p(\bar t,m^\ast)+\Delta\\
        &\hspace{4.1em}+(1-p(\bar t,m^\ast))
        \frac{\beta m^\ast}{1+\beta m^\ast}
        \Biggr].
    \end{align*}
    This is~\eqref{eq:foc_e}. The bracket is $H(\bar t,m^\ast)+\Delta$.
    
    For the second-order condition,
    \[
        H_m(t,m)
        =\frac{\beta E(t,m)(2+\beta m)}{(1+\beta m)^2}>0.
    \]
    Because $m_e^\ast<0$,
    \[
        \frac{d^2U_V}{de^2}
        =-k_e+\lambda F H_m(\bar t,m^\ast)m_e^\ast<-k_e<0.
    \]
    Thus the interior solution is the unique vendor optimum; if its unconstrained value reaches~$\theta$, the optimum is the full-mitigation corner.
\end{proof}

\begin{remark}[Vendor participation]
\label{rem:participation}
    The vendor meets the floor when
    \[
        B>\frac{k_t}{2}\bar t^2+\frac{k_e}{2}(e^\ast)^2
        +\lambda F(\theta-e^\ast)[p(\bar t,m^\ast)+\Delta].
    \]
    This holds for sufficiently large adoption benefit~$B$. Otherwise the model's conditional-on-adoption equilibrium is not reached.
\end{remark}

\subsection{Uniqueness and Comparative Responses}

For a fixed institutional environment, substitute the deployer response $m^\ast(e)$ into the vendor condition:
\[
    \Phi(e;\Delta,s)=A[H(\bar t,m^\ast(e,s))+\Delta].
\]
Equation~\eqref{eq:e_star} is $e^\ast=\min\{\theta,\Phi(e^\ast)\}$. Since
\[
    \Phi_e=A H_m m_e^\ast<0,
\]
$\Phi$ is strictly decreasing while the identity is increasing. Hence there is at most one interior crossing. If $\Phi(\theta^-)<\theta$, continuity yields one interior equilibrium; if $\Phi(\theta^-)\ge\theta$, the unique equilibrium is the corner $e^\ast=\theta$, $m^\ast=0$.

Define
\[
    D\equiv1-AH_m m_e^\ast>1.
\]
The equilibrium response to independent audit capacity is
\begin{equation}
\label{eq:delta-response}
    \frac{de^\ast}{d\Delta}=\frac{A}{D}>0.
\end{equation}
The direct fixed-monitoring effect is~$A$; the fall in monitoring attenuates but never reverses it.

At fixed mitigation,~\eqref{eq:ms-partial} gives $\left.m_s^\ast\right|_{e,t}<0$. Therefore
\begin{equation}
\label{eq:s-response}
    \frac{de^\ast}{ds}
    =\frac{A H_m}{D}\left.m_s^\ast\right|_{e,t}<0.
\end{equation}
The total monitoring response is also negative. One direct proof substitutes
$e=A[H(\bar t,m)+\Delta]$ into~\eqref{eq:m_star} and defines
\[
    \begin{aligned}
        \Psi(m,s)=k_m m
        &-\kappa\!\left[\theta-A(H(\bar t,m)+\Delta)\right]\\
        &\quad\times\frac{\delta}{1+s}\beta E(\bar t,m).
    \end{aligned}
\]
Here $\Psi_m>0$ and $\Psi_s>0$, so $dm^\ast/ds=-\Psi_s/\Psi_m<0$.

\begin{proof}[Proof of Proposition~\ref{prop:proxy}]
    As $s\to\infty$, $\deff\to0$. Equation~\eqref{eq:m_star} then forces $m^\ast\to0$, so $H(\bar t,m^\ast)\to H(\bar t,0)=1-\exp(-\alpha\bar t)$. Under the stated weak-enforcement condition, the limiting mitigation is below first best. Equation~\eqref{eq:s-response} shows that mitigation is decreasing in~$s$, so by continuity it remains below first best for all sufficiently large~$s$. When $\bar t$ is small, $1-\exp(-\alpha\bar t)\simeq\alpha\bar t$, making the limiting mitigation correspondingly small.
\end{proof}

\begin{proof}[Proof of Proposition~\ref{prop:account}]
    Part~(a) follows from~\eqref{eq:delta-response}. Part~(b) follows from~\eqref{eq:s-response}. For part~(c), Theorem~\ref{thm:audit} supplies a finite audit threshold for every~$s$; lowering~$s$ raises $H$ through monitoring, weakly reducing the audit intensity needed to reach first best.
\end{proof}

\section{Audit Rights and Other Policy Proofs}
\label{app:policyproofs}

\begin{proof}[Proof of Theorem~\ref{thm:audit}]
    Fix~$s$. Let $m^{FB}(s)$ be the deployer best response when mitigation equals~$e^{FB}$; if $e^{FB}=\theta$, then $m^{FB}=0$. The exact smallest audit boost that makes the vendor weakly prefer first-best mitigation is
    \begin{equation}
    \label{eq:exact-threshold}
        \bar\Delta(s)=
        \max\!\left\{0,\frac{k_e e^{FB}}{\lambda F}-H(\bar t,m^{FB}(s))\right\}.
    \end{equation}
    This threshold is finite. Since $H\ge0$,
    \[
        \bar\Delta(s)\le\frac{k_e e^{FB}}{\lambda F}
    \]
    for every~$s$, yielding the uniform sufficient condition in the theorem. At any proxy-compliance baseline $e^\ast<e^{FB}$, reaching $e^{FB}$ strictly reduces expected harm. In the worked-example corner $e^{FB}=\theta$, the exact threshold is independent of~$s$:
    \[
        \bar\Delta=\frac{k_e\theta}{\lambda F}-[1-\exp(-\alpha\bar t)]
        =0.418731\ldots .
    \]
\end{proof}

\begin{proof}[Proof of Proposition~\ref{prop:portability}]
    Equation~\eqref{eq:s-response} gives $de^\ast/ds<0$. The m-only equation~$\Psi(m,s)=0$ above gives $dm^\ast/ds<0$. Thus reducing~$s$ raises monitoring and mitigation. Since expected harm is $\theta-e^\ast$ in the interior, it falls.
\end{proof}

\begin{proof}[Proof of Proposition~\ref{prop:incident}]
    Let $z_R=\mu\rho$. The reporting regime shifts the vendor condition to
    \[
        e=\Phi(e;\Delta,s)+\frac{z_R}{k_e}.
    \]
    Implicit differentiation yields
    \[
        \frac{de_R^\ast}{dz_R}
        =\frac{1/k_e}{1-AH_m m_e^\ast}>0.
    \]
    The monitoring response attenuates the direct shift but cannot reverse it. The result holds until the full-mitigation corner is reached.
\end{proof}

\begin{proof}[Proof of Proposition~\ref{prop:liability}]
    Let $z_O=\lambda_OF_O$. The outcome-linked term enters the vendor condition additively and independently of $p(t,m)$:
    \[
        e=\Phi(e;\Delta,s)+\frac{z_O}{k_e}.
    \]
    The same implicit derivative as above gives $de_L^\ast/dz_O>0$. The auditability derivative remains negative because $z_O$ does not depend on~$t$, so $t^\ast=\bar t$.
\end{proof}

\section{Worked-Example Calculations}
\label{app:worked}

The worked example uses
\[
    \begin{aligned}
        \theta&=0.6,\qquad k_t=k_e=k_m=1,\qquad \bar t=0.2,\\
        \alpha&=\beta=\kappa=\delta=\lambda F=1.
    \end{aligned}
\]
For an interior setting, solve
\begin{align*}
    m^\ast&=(0.6-e^\ast)\frac{1}{1+s}\exp(-0.2-m^\ast),\\
    e^\ast&=1-\frac{\exp(-0.2-m^\ast)}{1+m^\ast}+\Delta.
\end{align*}
The resulting values before rounding are:
\begin{table}[t]
    \centering
    \caption{Unrounded worked-example solutions.}
    \label{tab:unrounded}
    \scriptsize
    \begin{tabular}{@{}lrrrr@{}}
        \toprule
        Setting & $m^\ast$ & $p(\bar t,m^\ast)$ & $e^\ast$ & $\E[h]$\\
        \midrule
        High lock-in & 0.029573 & 0.205127 & 0.227958 & 0.372042\\
        Portability only & 0.100822 & 0.259791 & 0.327585 & 0.272415\\
        Combined reform & 0 & 0.181269 & 0.600000 & 0\\
        \bottomrule
    \end{tabular}
\end{table}

In the combined setting, even the limiting marginal exposure at $m=0$ is $H(\bar t,0)+\Delta=0.181269+1>0.6$, so mitigation clamps at~$\theta$. With $q=0$, Equation~\eqref{eq:foc_m} then gives $m^\ast=0$. Portability alone reduces expected harm by approximately $26.8\%$ relative to the high-lock-in baseline.

Figure~\ref{fig:regimes} is generated from the same equations. At a target mitigation~$e_0<\theta$, monitoring has the Lambert-$W$ representation
\[
    m(e_0,s)=\frac{1}{\beta}
    W\!\left(
    \frac{\kappa(\theta-e_0)\delta\beta^2e^{-\alpha\bar t}}
    {k_m(1+s)}
    \right),
\]
and the audit boost required to support that target is
\[
    \Delta(e_0,s)=
    \max\!\left\{0,\frac{k_e e_0}{\lambda F}-H(\bar t,m(e_0,s))\right\}.
\]
The plotted contours use $e_0=e^{FB}/2$ and $e_0=e^{FB}$.

\section{Comparative Statics}
\label{app:comparative}

The denominator~$D=1-AH_m m_e^\ast$ is positive. Besides~\eqref{eq:delta-response} and~\eqref{eq:s-response}, the following derivatives hold in the interior:
\begin{align*}
    \frac{de^\ast}{d\bar t}
    &=\frac{A}{D}\,
    \frac{\alpha E(\bar t,m^\ast)}{(1+\beta m^\ast)^3}>0,\\
    \frac{de^\ast}{dk_e}
    &=-\frac{e^\ast}{k_eD}<0,\\
    \frac{de^\ast}{d(\lambda F)}
    &=\frac{e^\ast}{\lambda F D}>0,\\
    \frac{de^\ast}{d\kappa}
    &=\frac{AH_m}{D}\,
    \frac{m^\ast}{\kappa(1+\beta m^\ast)}>0.
\end{align*}
For the $\bar t$ derivative, the direct increase in detection dominates the induced reduction in monitoring:
\[
    \frac{dH}{d\bar t}\Big|_e
    =H_t+H_m m_t^\ast
    =\frac{\alpha E}{(1+\beta m^\ast)^3}>0.
\]

\begin{table}[t]
    \centering
    \caption{Total equilibrium effects in the interior. A dash denotes a decrease, a plus an increase, and zero no
    change. Monitoring effects include the endogenous mitigation response.}
    \label{tab:compstat}
    \small
    \begin{tabular}{@{}lcccc@{}}
        \toprule
        Parameter increase & $m^\ast$ & $t^\ast$ & $e^\ast$ & $\E[h]$\\
        \midrule
        $s$ (lock-in) & $-$ & $0$ & $-$ & $+$\\
        $\Delta$ (independent audit) & $-$ & $0$ & $+$ & $-$\\
        $\lambda F$ (evidence-based exposure) & $-$ & $0$ & $+$ & $-$\\
        $\bar t$ (procurement floor) & $-$ & $+$ & $+$ & $-$\\
        $k_e$ (mitigation cost) & $+$ & $0$ & $-$ & $+$\\
        $\kappa$ (deployer harm exposure) & $+$ & $0$ & $+$ & $-$\\
        \bottomrule
    \end{tabular}
\end{table}

The negative monitoring response to stronger audit or enforcement is not a policy failure: higher mitigation lowers residual harm, reducing the deployer's need to monitor.

\section{Regulatory Mapping in Greater Detail}
\label{app:regulation}

\paragraph{AI Act conformity and model oversight.}
Article~43 of Regulation (EU) 2024/1689 specifies conformity-assessment routes for high-risk systems; Annex~VI sets out internal control, while product-linked systems may follow sectoral third-party procedures. Articles~53 and~55 govern general-purpose model providers and models with systemic risk. Articles~91--93 grant the AI Office powers to request documentation and information, conduct evaluations, and require measures. The model distinguishes provider auditability~$t$ from genuinely independent capacity~$\Delta$: adding documentation to a provider-controlled process need not equal independent verification.

\paragraph{Post-market monitoring and incident reporting.}
Articles~72--73 create post-market monitoring and serious-incident reporting obligations for high-risk systems. Their effective model counterpart is~$\mu\rho$, not the nominal duty alone. Coverage, causal attribution, timeliness, regulator capacity, and sanctions for non-compliance determine~$\rho$ and~$\mu$.

\paragraph{Researcher access and portability analogies.}
DSA Article~40 provides vetted-researcher access for VLOPs and very large online search engines (VLOSEs) and therefore illustrates an access right that can raise independent detection. DMA Articles~5--6 illustrate data-portability and interoperability obligations that can lower switching costs. Neither regime maps mechanically onto foundation-model procurement; they identify the institutional form of the levers.

\paragraph{Liability.}
The withdrawn AI Liability Directive proposal would have eased proof through disclosure and rebuttable presumptions in fault-based claims. It was not automatic liability for every adverse outcome. Directive (EU) 2024/2853 modernizes strict product liability and expressly encompasses software; however, claimants still need to prove a defective product, covered damage, and causation. Proposition~\ref{prop:liability} should therefore be read as a benchmark for exposure that is less dependent on the paper's separate evidence channel, not as a literal description of current EU law.

\section{Operationalization and Empirical Designs}
\label{app:empirical}

\begin{table*}[t]
    \centering
    \caption{Illustrative measures for the model's institutional parameters.}
    \label{tab:measures}
    \small
    \begin{tabular}{@{}p{0.12\textwidth}p{0.34\textwidth}p{0.43\textwidth}@{}}
        \toprule
        Construct & Candidate measures & Candidate data sources or settings\\
        \midrule
        Auditability $t$ &
        contractual audit clauses; API rate and logging terms; reproducibility permissions; model-card and system-card
        completeness; publication restrictions &
        procurement contracts; platform terms; model documentation; DSA audit reports; researcher-access scorecards\\
        Independent capacity $\Delta$ &
        regulator testing rights; third-party access independent of provider permission; evaluation staffing and budgets;
        ability to publish adverse findings &
        AI Office and national authority records; AI Safety Institute programs; DSA Article~40 access decisions\\
        Lock-in $s$ &
        provider-specific fine-tuning artifacts; prompt and tool wrappers; data-egress and retraining costs; switching and
        re-tender rates &
        machine-learning operations (MLOps) telemetry; procurement records; vendor-migration projects; surveys of deployers\\
        Monitoring $m$ &
        continuous evaluation frequency; red-team budget; incident triage; external audit expenditure &
        internal governance records; audit contracts; safety-case updates; regulated-sector supervision\\
        Mitigation $e$ &
        patching and retraining effort; safety engineering headcount; post-incident remediation; independently measured
        performance improvements &
        change logs; technical reports; incident investigations; independent benchmarks\\
        Reporting $\rho$ &
        report completeness, timeliness, scope, and match between internal incidents and regulator filings &
        AI Act Article~73 reports when available; sectoral incident registers; AI Incident Database\\
        \bottomrule
    \end{tabular}
\end{table*}

Three empirical designs are particularly natural. First, procurement-threshold changes can be studied with difference-in-differences or event-study designs when comparable sectors adopt requirements at different times. Second, provider migrations or interoperability mandates create settings for measuring whether reduced lock-in raises monitoring and remediation. Third, changes in independent access rights can be linked to the gap between vendor-reported and independently reproduced performance. In each case, underlying sector risk and institutional capacity are likely confounders; the model identifies directional hypotheses, not a ready-made causal instrument.

\section{Modeling Choices and Extensions}
\label{app:scope}

\paragraph{Common-knowledge risk.}
Conditioning on~$\theta$ isolates evidence scarcity from adverse selection. A private-information extension would let the vendor signal risk through $(t,e)$ and could add pooling or separating equilibria, but it is not required for the proxy-compliance mechanism.

\paragraph{Static, single-vendor structure.}
Repeated interaction could make reputation an additional enforcement channel or, conversely, institutionalize symbolic compliance. A multi-vendor model would endogenize switching opportunities and market concentration. The present~$s$ summarizes these forces rather than deriving them.

\paragraph{Harm technology.}
The linear specification makes the monitoring response and policy thresholds transparent. More general decreasing harm functions preserve the mechanism when the vendor objective remains concave and the deployer best response is well-behaved, but the threshold formulas change.

\paragraph{Non-market evidence producers.}
Journalists, civil society organizations, and academic researchers can be represented as contributors to~$\Delta$. Their access, publication incentives, and vulnerability to legal or contractual restriction are themselves strategic objects and deserve a separate model.

\paragraph{Scaling and political economy.}
The paper's title identifies the application domain. Scale, compute concentration, and regulatory capacity are not endogenous primitives. They motivate why audit access and switching costs may be consequential, but the formal results should not be read as comparative statics in model scale.

\section{Welfare Effects}
\label{app:welfare}

Treating penalties as transfers and abstracting from policy-administration costs, define welfare conditional on adoption as
\begin{equation}
\label{eq:welfare}
    W=V-(\theta-e^\ast)-\frac{k_e}{2}(e^\ast)^2
    -\frac{k_t}{2}\bar t^2-\frac{k_m}{2}(m^\ast)^2.
\end{equation}
Because the deployer best response depends on policy only through equilibrium mitigation for an audit change, $m_\Delta^\ast=m_e^\ast e_\Delta^\ast$. Therefore
\[
    \frac{dW}{d\Delta}
    =\left[(1-k_e e^\ast)-k_m m^\ast m_e^\ast\right]
    \frac{de^\ast}{d\Delta}.
\]
When $e^\ast<e^{FB}$, $1-k_e e^\ast>0$; because $m_e^\ast<0$ and $de^\ast/d\Delta>0$, the audit effect is strictly positive. Stronger audit capacity raises beneficial mitigation and reduces the monitoring cost required to manage residual harm.

For switching costs,
\[
    \frac{dW}{ds}
    =(1-k_e e^\ast)\frac{de^\ast}{ds}
    -k_m m^\ast\frac{dm^\ast}{ds}.
\]
Both derivatives are negative. A marginal reduction in~$s$ improves welfare precisely when
\[
    (1-k_e e^\ast)\left|\frac{de^\ast}{ds}\right|
    >
    k_m m^\ast\left|\frac{dm^\ast}{ds}\right|.
\]
The left side is the net benefit from additional mitigation; the right side is the extra monitoring resource cost. This condition replaces the stronger and generally unjustified claim that portability is always welfare-improving.

\section*{Ethical Statement}
This paper develops a stylized institutional model and uses no human participants, personal data, or deployed AI systems. Its principal adverse-impact risk is overgeneralization: the comparative statics do not establish that any specific firm, sector, or legal regime necessarily exhibits proxy compliance, nor that audit, reporting, portability, or liability is sufficient on its own. We therefore separate the model's benchmark mechanisms from claims about current law and stress that implementation depends on audit independence, due process, confidentiality safeguards, civil-society access, distributional effects, and regulator capacity.

\singlespacing
\printbibliography

\end{document}